\documentclass[dvips,generic,twoside,noinfoline]{imsart}

\usepackage[OT1]{fontenc}
\usepackage{graphicx}
\usepackage{amssymb}
\usepackage{amsthm}
\usepackage{amsmath}
\usepackage{subfig}
\usepackage{natbib}

\newcommand{\bigoh}{\mathcal{O}}
\newcommand{\im}{\mathop{\rm Im}}
\newcommand{\Prob}{\mathop{\rm Prob}}
\newtheorem{observe}{Observation}

\newtheorem{theorem}{Theorem}[section]
\newtheorem{lemma}[theorem]{Lemma}

\newtheorem{remark1}[theorem]{Remark}

\newenvironment{remark}{\begin{remark1} \rm}{\end{remark1}}

\startlocaldefs
\def\phi{\varphi}
\def\epsilon{\varepsilon}
\endlocaldefs

\begin{document}

\begin{frontmatter}

\title{Computing the asymptotic power of a Euclidean-distance test
       for goodness-of-fit\thanksref{tt}}
\runtitle{Computing the asymptotic power of a Euclidean-distance test
          for goodness-of-fit}

\author{\fnms{William} \snm{Perkins}\ead[label=e1]{perkins@math.gatech.edu}}
\address{School of Mathematics\\Georgia Institute of Technology\\
         686 Cherry St.\\Atlanta, GA 30332-0160\\\printead{e1}}

\author{\fnms{Gary} \snm{Simon}\ead[label=e4]{gsimon@stern.nyu.edu}}
\address{IOMS Department\\Stern School of Business\\NYU\\44 West 4th St.\\
         New York, NY 10012\\\printead{e4}}

\author{and\\}

\author{\fnms{Mark} \snm{Tygert}\ead[label=e2]{tygert@aya.yale.edu}}
\address{Courant Institute of Mathematical Sciences\\NYU\\251 Mercer St.\\
         New York, NY 10012\\\printead{e2}}


\thankstext{tt}{Supported in part
by NSF Grant OISE-0730136, an NSF Postdoctoral Research Fellowship,
a DARPA Young Faculty Award, and an Alfred P. Sloan Research Fellowship.}

\runauthor{W. Perkins, G. Simon, and M. Tygert}

\begin{abstract}
A natural (yet unconventional) test for goodness-of-fit measures
the discrepancy between the model and empirical distributions
via their Euclidean distance (or, equivalently, via its square).
The present paper characterizes the statistical power of such a test
against a family of alternative distributions,
in the limit that the number of observations is large,
with every alternative departing from the model in the same direction.
Specifically, the paper provides an efficient numerical method for evaluating
the cumulative distribution function (cdf) of the square
of the Euclidean distance between the model and empirical distributions
under the alternatives, in the limit that the number of observations is large.
The paper illustrates the scheme by plotting the asymptotic power
(as a function of the significance level) for several examples.
\end{abstract}



\tableofcontents

\end{frontmatter}

\section{Introduction}

Given $n$ observations, each falling in one of $m$ bins,
we would like to test if these observations are consistent with having arisen
as independent and identically distributed (i.i.d.)\ draws
from a specified probability distribution $p_0$
over the $m$ bins ($p_0$ is known as the ``model'').
A natural measure of the deviation between $p_0$ and the observations is
the square $x_a$ of the Euclidean distance
between the actually observed distribution of the draws
and the expected distribution $p_0$, that is,
\begin{equation}
x_a = \sum_{k=1}^m ((y_a)_k - (p_0)_k)^2,
\end{equation}
where $(y_a)_1$,~$(y_a)_2$, \dots, $(y_a)_m$ are the proportions
of the $n$ observations falling in bins $1$,~$2$, \dots, $m$, respectively.

The ``P-value'' is then defined to be the probability that $X_0 \ge x_a$,
where $X_0$ would be the same as $x_a$, but constructed from $n$ draws
that definitely are taken i.i.d.\ from $p_0$, that is,
\begin{equation}
X_0 = \sum_{k=1}^m ((Y_0)_k - (p_0)_k)^2,
\end{equation}
where $(Y_0)_1$,~$(Y_0)_2$, \dots, $(Y_0)_m$ are the proportions
of $n$ i.i.d.\ draws from $p_0$ falling in bins $1$,~$2$, \dots, $m$,
respectively. When calculating the P-value
--- the probability that $X_0 \ge x_a$ --- we view $X_0$ as a random variable
while viewing $x_a$ as a fixed number.
If the P-value is small, then we can be confident that the observed draws
were not taken i.i.d.\ from the model $p_0$.

To characterize the statistical power of the P-value based
on the Euclidean distance,
we consider $n$ i.i.d.\ draws from the alternative distribution
\begin{equation}
\label{initalt}
p_a = p_0 + a/\sqrt{n},
\end{equation}
where $a$ is a vector whose $m$ entries satisfy $\sum_{k=1}^m a_k = 0$.
We thus need to calculate the distribution of the square $X_a$
of the Euclidean distance,
\begin{equation}
\label{alt}
X_a = \sum_{k=1}^m ((Y_a)_k - (p_0)_k)^2,
\end{equation}
where $(Y_a)_1$,~$(Y_a)_2$, \dots, $(Y_a)_m$ are the proportions
of $n$ i.i.d.\ draws from $p_a$ falling in bins $1$,~$2$, \dots, $m$,
respectively.
Section~\ref{method} below provides an efficient method
for calculating the cumulative distribution function (cdf) of $n \cdot X_a$
in the limit that the number $n$ of draws is large.
Section~\ref{scheme} below then describes how to use such a method
to plot the cdf of the P-values;
this cdf is the same as the statistical power function
of the hypothesis test based on the Euclidean distance
(as a function of the significance level).
Presenting this method is the principal purpose of the present paper,
complementing the earlier discussions of~\cite{perkins-tygert-ward}
and~\cite{perkins-tygert-ward3},
which compare the Euclidean distance with classical statistics
such as $\chi^2$, the log--likelihood-ratio $G^2$, and other members
of the Cressie-Read power-divergence family;
\cite{perkins-tygert-ward} and~\cite{perkins-tygert-ward3}
review the classical statistics and provide detailed comparisons.

As reviewed, for example, by~\cite{kendall-stuart-ord-arnold} and~\cite{rao},
$m \cdot n \cdot X_a$ defined in~(\ref{alt}) converges in distribution
to a noncentral $\chi^2$ in the limit that the number $n$ of draws is large,
when the model $p_0$ is a uniform distribution.
When $p_0$ is nonuniform, $m \cdot n \cdot X_a$ converges in distribution
to the sum of the squares of independent Gaussian random variables
in the limit that the number $n$ of draws is large,
as shown by~\cite{moore-spruill} and reviewed in Section~\ref{linalg} below.
Section~\ref{representation} provides integral representations
for the cdf of the sum of the squares of independent Gaussian random variables
and applies suitable quadratures for their numerical evaluation.
Section~\ref{method} summarizes the numerical method obtained
by combining Sections~\ref{linalg} and~\ref{representation}.
Section~\ref{scheme} summarizes a scheme
for plotting the asymptotic power (as a function of the significance level)
using the method of Section~\ref{method}.
Section~\ref{examples} illustrates the methods via several numerical examples.

The extension to models with nuisance parameters is straightforward,
following~\cite{perkins-tygert-ward2};
the present paper focuses on the simpler case in which the model $p_0$
is a single, fully specified probability distribution.

\section{Preliminaries}
\label{linalg}

This section states Theorem~\ref{ms}, which is a special case
of Theorem~4.2 of~\cite{moore-spruill}. Before stating the theorem,
we need to set up some notation. The set-up amounts to an algorithm
for computing the real numbers $\sigma_1$,~$\sigma_2$, \dots, $\sigma_{m-1}$
and $\zeta_1$,~$\zeta_2$, \dots, $\zeta_{m-1}$ used in Theorem~\ref{ms},
where $m$ is an integer greater than 1.

First, we aim to define the positive real numbers
$\sigma_1$,~$\sigma_2$, \dots, $\sigma_{m-1}$,
given any $m \times 1$ vector $p_0$ whose entries are all positive.
We define $D$ to be the diagonal $m \times m$ matrix
\begin{equation}
D_{j,k} = \left\{ \begin{array}{rl}
                  \frac{1}{(p_0)_j}, & j = k \phantom{\Bigm|}\\
                  0, & j \ne k
                  \end{array} \right.
\end{equation}
for $j,k = 1$,~$2$, \dots, $m$.
We define $H$ to be the $m \times m$ matrix
\begin{equation}
H_{j,k} = \left\{ \begin{array}{rl}
                  1 - \frac{1}{m}, & j = k \phantom{\Bigm|}\\
                  -\frac{1}{m}, & j \ne k
                  \end{array} \right.
\end{equation}
for $j,k = 1$,~$2$, \dots, $m$.
Note that $H$ is an orthogonal projector.
We define $B = HDH$, so that $B$ is the self-adjoint $m \times m$ matrix
\begin{equation}
\label{diagonalizer}
B_{j,k} = \left\{ \begin{array}{rl}
          \frac{1}{(p_0)_j}
          - \frac{1}{m} \Bigl( \frac{1}{(p_0)_j}+\frac{1}{(p_0)_k} \Bigr)
          + \frac{1}{m^2} \sum_{l=1}^m \frac{1}{(p_0)_l}, & j = k
                                                            \phantom{\biggm|}\\
          - \frac{1}{m} \Bigl( \frac{1}{(p_0)_j}+\frac{1}{(p_0)_k} \Bigr)
          + \frac{1}{m^2} \sum_{l=1}^m \frac{1}{(p_0)_l}, & j \ne k
          \end{array} \right.
\end{equation}
for $j,k = 1$,~$2$, \dots, $m$.
As a self-adjoint matrix whose rank is $m-1$
(after all, $B = HDH$, $H$ is an orthogonal projector whose rank is $m-1$,
and $D$ is a full-rank diagonal matrix), $B$ given in~(\ref{diagonalizer})
has an eigendecomposition
\begin{equation}
\label{eigendecomposition}
B = Q \Lambda Q^{\top},
\end{equation}
where $Q$ is a real unitary $m \times m$ matrix
and $\Lambda$ is a diagonal $m \times m$ matrix
such that $\Lambda_{m,m} = 0$.
Finally, we define the positive real numbers
$\sigma_1$,~$\sigma_2$, \dots, $\sigma_{m-1}$ via the formula
\begin{equation}
\label{stddev}
\sigma_k^2 = 1/\Lambda_{k,k}
\end{equation}
for $k = 1$,~$2$, \dots, $m-1$,
where $\Lambda_{1,1}$,~$\Lambda_{2,2}$, \dots, $\Lambda_{m,m}$
are the diagonal entries of $\Lambda$
from the eigendecomposition~(\ref{eigendecomposition}).

Next, we define the real numbers $\zeta_1$,~$\zeta_2$, \dots, $\zeta_{m-1}$,
given both $p_0$ and an $m \times 1$ vector $a$
such that $\sum_{k=1}^m a_k = 0$.
We define the $(m-1) \times 1$ vector
\begin{equation}
\label{almostthere}
\eta = \tilde{Q}^{\top} a,
\end{equation}
where $\tilde{Q}$ is the leftmost $m \times (m-1)$ block of $Q$
from the eigendecomposition~(\ref{eigendecomposition}),
that is, $\tilde{Q}$ is the same as $Q$ after deleting the last column of $Q$.
We can then define the real numbers $\zeta_1$,~$\zeta_2$, \dots, $\zeta_{m-1}$
via the formula
\begin{equation}
\label{offset}
\zeta_k = \eta_k/\sigma_k
\end{equation}
for $k = 1$,~$2$, \dots, $m-1$,
where $\eta$ is defined in~(\ref{almostthere})
and $\sigma$ is defined in~(\ref{stddev}).

With this notation, we can state the following special case
of Theorem~4.2 of~\cite{moore-spruill}.

\begin{theorem}
\label{ms}
Suppose that $m$ is an integer greater than one,
$p_0$ is a probability distribution over $m$ bins
(that is, $p_0$ is an $m \times 1$ vector whose entries are all positive
and $\sum_{k=1}^m (p_0)_k = 1$),
$a$ is an $m \times 1$ vector such that $\sum_{k=1}^m a_k = 0$,
and $(Y_n)_1$,~$(Y_n)_2$, \dots, $(Y_n)_m$ are the proportions of draws falling
in bins $1$,~$2$, \dots, $m$, respectively, out of a total of $n$ i.i.d.\ draws
from the probability distribution
\begin{equation}
p_a = p_0 + a/\sqrt{n}.
\end{equation}
Suppose further that $X_n$ is the random variable
\begin{equation}
X_n = n \sum_{k=1}^m ((Y_n)_k - (p_0)_k)^2.
\end{equation}

Then, $X_n$ converges in distribution to the random variable
\begin{equation}
\label{limiting}
X_{\infty} = \sum_{k=1}^{m-1} \sigma_k^2 \, (Z_k + \zeta_k)^2
\end{equation}
as $n$ becomes large,
where $Z_1$,~$Z_2$, \dots, $Z_{m-1}$ are i.i.d.\ Gaussian random variables
of zero mean and unit variance,
$\sigma_1$,~$\sigma_2$, \dots, $\sigma_{m-1}$ are the positive real numbers
defined in~(\ref{stddev}),
and $\zeta_1$,~$\zeta_2$, \dots, $\zeta_{m-1}$ are the real numbers
defined in~(\ref{offset}).
The values of $\sigma_1$,~$\sigma_2$, \dots, $\sigma_{m-1}$ do not depend
on the vector $a$; the values of $\zeta_1$,~$\zeta_2$, \dots, $\zeta_{m-1}$
do depend on $a$.
\end{theorem}

\begin{remark}
\label{faster}
The $m \times m$ matrix $B$ defined in~(\ref{diagonalizer}) is the sum
of a diagonal matrix and a low-rank matrix.
The methods of~\cite{gu-eisenstat94,gu-eisenstat95}
for computing the eigenvalues of such a matrix~$B$ 
and computing the result of applying $Q^{\top}$ from~(\ref{eigendecomposition})
to an arbitrary vector require only
either $\bigoh(m^2)$ or $\bigoh(m \log(m))$ floating-point operations.
The $\bigoh(m^2)$ methods of~\cite{gu-eisenstat94,gu-eisenstat95} are usually
more efficient than the $\bigoh(m \log(m))$ method of~\cite{gu-eisenstat95},
unless $m$ is impractically large.
\end{remark}

\section{Integral representations}
\label{representation}

This section describes efficient algorithms for evaluating
the cdf of the sum~(\ref{limiting}) of the squares
of independent Gaussian random variables.
The bibliography of~\cite{duchesne-micheaux} gives references
to possible alternatives to the methods of the present section.
Our principal tool is the following theorem,
representing the cdf as an integral suitable for evaluation via quadratures
(see, for example, Remark~\ref{quadratures} below);
the theorem expresses formula~7 of~\cite{rice}
in the same form as formula~8 of~\cite{perkins-tygert-ward}.

\begin{theorem}
Suppose that $\ell$ is a positive integer,
$Z_1$,~$Z_2$, \dots, $Z_{\ell}$ are i.i.d.\ Gaussian random variables
of zero mean and unit variance,
and $\sigma_1$,~$\sigma_2$, \dots, $\sigma_{\ell}$
and $\zeta_1$,~$\zeta_2$, \dots, $\zeta_{\ell}$
are real numbers.
Suppose in addition that $X$ is the random variable
\begin{equation}
\label{summed}
X = \sum_{k=1}^{\ell} \sigma_k^2 \, (Z_k+\zeta_k)^2.
\end{equation}

Then, the cdf $F$ of $X$ is
\begin{equation}
\label{contoured}
F(x) = \int_0^\infty \im\left(
       \frac{e^{1-y} \, e^{iy\sqrt{\ell}} \,
        \prod_{k=1}^{\ell} e^{\zeta_k^2 (1-w_k(y))/(2w_k(y))}}
       {\pi \, \bigl( y - \frac{1}{1-i\sqrt{\ell}} \bigr)
        \prod_{k=1}^{\ell} \sqrt{w_k(y)}}
       \right) \, dy
\end{equation}
for any positive real number $x$,
where
\begin{equation}
\label{specfun}
w_k(y) = 1-2(y-1)\sigma_k^2/x+2iy\sigma_k^2\sqrt{\ell}/x,
\end{equation}
and $F(x) = 0$ for any nonpositive real number $x$.
The square roots in~(\ref{contoured}) denote the principal branch,
and $\,\im$ takes the imaginary part.
\end{theorem}

\begin{remark}
\label{quadratures}
An efficient means of evaluating~(\ref{contoured}) numerically
is to employ adaptive Gaussian quadratures;
see, for example, Section~4.7 of~\cite{press-teukolsky-vetterling-flannery}.
Good choices for the lowest orders of the quadratures used in the adaptive
Gaussian quadratures are 10 and 21, for double-precision accuracy.
\end{remark}

The remainder of the present section (particularly Remark~\ref{stabrem})
discusses the numerical stability of the method of Remark~\ref{quadratures}
and recalls an alternative integral representation suitable for use
when the method of Remark~\ref{quadratures} is not guaranteed
to be numerically stable.
The following lemma, proven in Remark~3.2 of~\cite{perkins-tygert-ward},
ensures that the denominator in~(\ref{contoured}) is not too small.

\begin{lemma}
Suppose that $\ell$ is a positive integer,
and $r_1$,~$r_2$, \dots, $r_{\ell}$ and $y$ are positive real numbers.
Suppose further that (in parallel with formula~(\ref{specfun}) above)
\begin{equation}
\label{newspecfun}
w_k = 1-r_k(y-1)+r_kiy\sqrt{\ell}
\end{equation}
for $k = 1$,~$2$, \dots, $\ell$.

Then,
\begin{equation}
\label{denom}
\left|\prod_{k=1}^{\ell} \sqrt{w_k}\right| > e^{-1/4}.
\end{equation}
\end{lemma}

The following lemma ensures that the numerator in~(\ref{contoured})
is not too large, provided that $e^{\zeta_k^2/2}$ is not large.

\begin{lemma}
Suppose that $r$, $y$, and $\ell$ are positive real numbers and
(in parallel with formulae~(\ref{specfun}) and~(\ref{newspecfun}) above)
\begin{equation}
\label{contour}
w = 1-r(y-1)+riy\sqrt{\ell}.
\end{equation}

Then,
\begin{equation}
\label{bound}
\left|\frac{1-w}{w}\right| \le \sqrt{1 + \frac{1}{\ell}}.
\end{equation}
\end{lemma}

\begin{proof}
Defining
\begin{equation}
\label{defz}
z = \frac{1}{y}
\end{equation}
and
\begin{equation}
\label{defc}
c = 1 + \frac{1}{r},
\end{equation}
we obtain that
\begin{equation}
\label{identity}
\frac{1-w}{w} = -\frac{1-z-i\sqrt{\ell}}{1-cz-i\sqrt{\ell}}.
\end{equation}
It follows from~(\ref{identity}) that
\begin{equation}
\label{initbound}
\left|\frac{1-w}{w}\right|^2 = \frac{(1-z)^2 + \ell}{(1-cz)^2 + \ell}.
\end{equation}

It follows from~(\ref{defz}) that $z \ge 0$
and from~(\ref{defc}) that $c \ge 1$, and hence
\begin{equation}
\label{inequal}
cz-1 \ge z-1.
\end{equation}
If $z \ge 1$, then~(\ref{inequal}) yields that
\begin{equation}
(cz-1)^2 \ge (z-1)^2,
\end{equation}
which in turn yields that
\begin{equation}
\label{case1}
\frac{(1-z)^2 + \ell}{(1-cz)^2 + \ell}
\le \frac{(1-z)^2 + \ell}{(1-z)^2 + \ell} = 1.
\end{equation}

If $z \le 1$, then (recalling that $z \ge 0$, too)
\begin{equation}
\label{case2}
\frac{(1-z)^2 + \ell}{(1-cz)^2 + \ell} \le \frac{(1-z)^2 + \ell}{\ell}
\le \frac{1+\ell}{\ell}.
\end{equation}

We see from~(\ref{case1}) and~(\ref{case2}) that, in all cases,
\begin{equation}
\label{cases}
\frac{(1-z)^2 + \ell}{(1-cz)^2 + \ell} \le 1 + \frac{1}{\ell}.
\end{equation}
Combining~(\ref{initbound}) and~(\ref{cases}) yields~(\ref{bound}).
\end{proof}

\begin{remark}
\label{stabrem}
The bound~(\ref{denom}) shows that the integrand in~(\ref{contoured})
is not too large for any nonnegative $y$,
provided that the numerator of~(\ref{contoured}) is not too large.
An upper bound on the numerator follows immediately from~(\ref{bound}):
\begin{equation}
\label{stability}
\left| \prod_{k=1}^{\ell} e^{\zeta_k^2 (1-w_k(y))/(2w_k(y))} \right|
\le \prod_{k=1}^{\ell} e^{\zeta_k^2 \sqrt{1+1/\ell}/2}.
\end{equation}
For any particular application,
we can check that the right-hand side of~(\ref{stability})
is not too many orders of magnitude in size,
guaranteeing that applying quadratures to the integral in~(\ref{contoured})
cannot lead to catastrophic cancellation in floating-point arithmetic.
Naturally, it is also possible to check on the magnitude of the integrand
in~(\ref{contoured}) during its numerical evaluation,
indicating even better numerical stability than guaranteed
by our {\it a priori} estimates.
See Theorem~\ref{altthm} and Remark~\ref{auxiliary} below
for an alternative integral representation
suitable for use when the right-hand side of~(\ref{stability}) is large.
\end{remark}

\begin{remark}
The bound in~(\ref{stability}) is quite pessimistic. In fact,
the real part of $(1-w_k(y))/(2w_k(y))$ is often nonpositive,
so that
\begin{equation}
\left|e^{\zeta_k^2 (1-w_k(y))/(2w_k(y))}\right| \le 1.
\end{equation}
\end{remark}

If the right-hand side of~(\ref{stability}) is large,
then we can use the method of~\cite{imhof}, \cite{davies}, and others,
applying numerical quadratures to the integral in the following theorem.
Please note that the integrand in the following theorem decays reasonably fast
when the right-hand side of~(\ref{stability}) is large. 

\begin{theorem}
\label{altthm}
Suppose that $\ell$ is a positive integer,
$Z_1$,~$Z_2$, \dots, $Z_{\ell}$ are i.i.d.\ Gaussian random variables
of zero mean and unit variance,
and $\sigma_1$,~$\sigma_2$, \dots, $\sigma_{\ell}$
and $\zeta_1$,~$\zeta_2$, \dots, $\zeta_{\ell}$
are real numbers.
Suppose in addition that $X$ is the random variable
\begin{equation}
\label{summed2}
X = \sum_{k=1}^{\ell} \sigma_k^2 \, (Z_k+\zeta_k)^2.
\end{equation}

Then, the cdf $F$ of $X$ is
\begin{equation}
\label{nocontour}
F(x) = \frac{1}{2}
     - \int_0^\infty \im\left(
       \frac{e^{-iy} \prod_{k=1}^{\ell} e^{\zeta_k^2 (1-v_k(y))/(2v_k(y))}}
            {\pi y \prod_{k=1}^{\ell} \sqrt{v_k(y)}}
       \right) \, dy
\end{equation}
for any positive real number $x$,
where
\begin{equation}
v_k(y) = 1-2iy\sigma_k^2/x,
\end{equation}
and $F(x) = 0$ for any nonpositive real number $x$.
The square roots in~(\ref{nocontour}) denote the principal branch,
and $\,\im$ takes the imaginary part.
\end{theorem}

\begin{remark}
\label{auxiliary}
The integrand in~(\ref{nocontour}) is not too large
(except for values of $y$ that are closer to 0
than are typical quadrature nodes),
since the real part of $(1-v_k(y))/(2v_k(y))$ is always nonpositive, so that
\begin{equation}
\left|e^{\zeta_k^2 (1-v_k(y))/(2v_k(y))}\right| \le 1.
\end{equation}
Moreover, the numerator in~(\ref{nocontour}) decays reasonably fast
(it is sub-Gaussian) when the right-hand side of~(\ref{stability}) is large. 
\end{remark}

\section{Numerical method}
\label{method}

Combining Sections~\ref{linalg} and~\ref{representation} yields
an efficient method for calculating the cdf $F$
of $n$ times the square of the Euclidean distance
between the model and empirical distributions,
in the limit that $n$ is large,
when the $n$ observed draws are taken i.i.d.\ from an alternative distribution
$p_a = p_0 + a/\sqrt{n}$
(as always, $p_0$ is the model --- a probability distribution over $m$ bins ---
and $a$ is a vector whose $m$ entries satisfy $\sum_{k=1}^m a_k = 0$).
Indeed, Theorem~\ref{ms} shows that the desired $F$ is the same as that
in~(\ref{contoured}) and~(\ref{nocontour}),
with the real numbers $\sigma_1$,~$\sigma_2$, \dots, $\sigma_{\ell}$
and $\zeta_1$,~$\zeta_2$, \dots, $\zeta_{\ell}$
calculated as detailed in Section~\ref{linalg} (identifying $\ell = m-1$).
Remark~\ref{quadratures} describes an efficient means
of evaluating $F(x)$ in~(\ref{contoured})
that is numerically stable when the right-hand side
of~(\ref{stability}) is not too many orders of magnitude in size.
When the right-hand side of~(\ref{stability}) is many orders of magnitude
in size, we can apply quadratures to the representation
of $F(x)$ in~(\ref{nocontour}) instead (see Remark~\ref{auxiliary}).

\section{Plotting the asymptotic statistical power}
\label{scheme}

Let us denote by $\pi$ the cdf of the P-values
for the Euclidean distance (or, equivalently, for any positive multiple
of the square of the Euclidean distance);
$\pi$ is also the statistical power function of the hypothesis test based
on the Euclidean distance (as a function of the significance level).
The method of Section~\ref{method} is sufficient for plotting $\pi$
in the limit that the number of draws is large.
Indeed, suppose that $X$ denotes $n$ times the square of the Euclidean distance
between the model and empirical distributions,
$F_0$ denotes the cdf for $X$
when taking $n$ draws i.i.d.\ from the model probability distribution $p_0$,
and $F_a$ denotes the cdf for $X$
when taking $n$ draws i.i.d.\ from $p_a = p_0 + a/\sqrt{n}$,
where $a$ is a vector whose $m$ entries satisfy $\sum_{k=1}^m a_k = 0$.
The P-value $P$ equals $1-F_0(X)$, in the limit that $n$ is large,
and then the cdf $\pi$ of the P-values for draws from $p_a$ is
\begin{multline}
\pi(1-F_0(x)) = \Prob\{P \le 1-F_0(x)\} = \Prob\{1-F_0(X) \le 1-F_0(x)\} \\
= \Prob\{X \ge x\} = 1-F_a(x)
\end{multline}
for any nonnegative real number $x$;
thus, the graph of all points $(\alpha,\pi(\alpha))$
with $\alpha$ ranging from $0$ to $1$
is the same as the graph of all points $(1-F_0(x),1-F_a(x))$
with $x$ ranging from $0$ to $\infty$,
in the limit that $n$ is large.
Section~\ref{method} describes how to evaluate $F_0(x)$ and $F_a(x)$
for any real number $x$, in the limit that the number $n$ of draws is large;
note that $F_0(x) = F_a(x)$ when the entries of $a$ are all zeros,
so the procedure of Section~\ref{method} can evaluate $F_0(x)$
as well as $F_a(x)$.
When the entries of $a$ are all zeros,
$\zeta_1 = \zeta_2 = \dots = \zeta_{\ell} = 0$
in the method of Section~\ref{method},
and then the right-hand side of~(\ref{stability}) is exactly 1.

\section{Numerical examples}
\label{examples}

This section illustrates the algorithms of the present paper
via several numerical examples.
As detailed in the subsections below,
we consider three examples for the model $p_0$
(as always, $p_0$ is a probability distribution over $m$ bins, that is,
a vector whose entries are all positive
and satisfy $\sum_{k=1}^m (p_0)_k = 1$),
taking $n$ i.i.d.\ draws from the alternative probability distribution
\begin{equation}
\label{alternative}
p_a = p_0 + a/\sqrt{n},
\end{equation}
where $a$ is a vector whose $m$ entries satisfy $\sum_{k=1}^m a_k = 0$
(the subsections below detail several examples for $a$).
Figure~\ref{fig} plots the cdf $\pi$
of the P-values for $n$ i.i.d.\ draws taken
from the alternative distribution $p_a$, when $n$ is large;
$\pi$ is also the statistical power function of the hypothesis test based
on the Euclidean distance (as a function of the significance level).
For each of the examples, Figure~\ref{fig} plots
the cdf $\pi$ both for $n =$ 1,000,000 draws
(computed via Monte-Carlo simulations) and in the limit that $n$ is large
(computed via the algorithms of the present paper); not surprisingly,
there is little difference between the plots
for $n =$ 1,000,000 and for the limit that $n$ is large.
The lines in Figure~\ref{fig} corresponding to $n =$ 1,000,000 draws
are colored green; the lines corresponding to the limit of large $n$ are black.

\begin{remark}
For each example, we computed the cdf $\pi$
for $n =$ 1,000,000 draws via 40,000 Monte-Carlo simulations.
A straightforward argument based on the binomial distribution,
detailed in Remark~3.4 of~\cite{perkins-tygert-ward3},
shows that the standard errors of the resulting estimates of the P-values $P$
are equal to $\sqrt{P(1-P)/40000} \le 0.0025$,
ensuring that the standard errors of the plotted abscissae $\alpha$
for the green points in Figure~\ref{fig} are approximately
$\sqrt{\alpha(1-\alpha)/40000} \le 0.0025$
(roughly the size of the radii of the plotted points).
\end{remark}

\begin{remark}
For each example, we plotted the cdf $\pi$
in the limit of a large number $n$ of draws
via the scheme of Section~\ref{scheme}.
Figure~\ref{fig} displays the points
$(\alpha,\pi(\alpha)) = (1-F_0(x),1-F_a(x))$
for the 10000 values $x = 1/2000$,~$2/2000$, \dots, $10000/2000$,
in the limit that the number $n$ of draws is large,
where $F_0(x)$ and $F_a(x)$ are defined in Section~\ref{scheme} and
computed to at least 6-digit accuracy via the method of Section~\ref{method}.
\end{remark}

Table~\ref{costs} summarizes computational costs
of the procedure described in Section~\ref{method}.
The headings of Table~\ref{costs} have the following meanings:
\begin{itemize}
\item $m$ is the number of bins in the probability distributions
      $p_0$ and $p_a$.
\item $q_0$ is the maximum number of quadrature nodes required
      in any of the 10000 evaluations of $F_0$ plotted in Figure~\ref{fig}
      (Section~\ref{scheme} defines $F_0$),
      using adaptive Gaussian quadratures as described
      in Remark~\ref{quadratures}.
\item $q_a$ is the maximum number of quadrature nodes required
      in any of the 10000 evaluations of $F_a$ plotted in Figure~\ref{fig}
      (Section~\ref{scheme} defines $F_a$),
      using adaptive Gaussian quadratures as described
      in Remark~\ref{quadratures}.
\item $t$ is the time in seconds required to perform the quadratures
      for both $F_0(x)$ and $F_a(x)$ at a single value of $x$,
      amortized over the 10000 pairs $(1-F_0(x),1-F_a(x))$ plotted
      in Figure~\ref{fig} (Section~\ref{scheme} defines $F_0$ and $F_a$).
\end{itemize}

\subsection{Uniform model}
\label{uniform}

For our first example, we take
\begin{equation}
(p_0)_k = 1/10
\end{equation}
for $k = 1$,~$2$, \dots, $10$, and take
\begin{equation}
a_k = (-1)^k/5
\end{equation}
for $k = 1$,~$2$, \dots, $10$.
The Euclidean distance is equivalent to the canonical $\chi^2$ statistic
for this example, since $p_0$ is a uniform distribution.

\subsection{Nonuniform model}
\label{nonuniform}

For our second example, we take
\begin{equation}
(p_0)_k = \left\{ \begin{array}{rl}
                  1/2, & k = 1 \\
                  1/198, & k = 2, 3, \dots, 100
                  \end{array} \right.
\end{equation}
for $k = 1$,~$2$, \dots, $100$, and take
\begin{equation}
a_k = \left\{ \begin{array}{rl}
              2/3, & k = 1 \\
              -2/297, & k = 2, 3, \dots, 100
              \end{array} \right.
\end{equation}
for $k = 1$,~$2$, \dots, $100$.

\subsection{Poisson model}
\label{Poisson}

For our third example, we take
\begin{equation}
(p_0)_k = e^{-3} \, 3^{k-1}/(k-1)!
\end{equation}
for $k = 1$,~$2$,~$3$, \dots, and take
\begin{equation}
a_k = \left\{ \begin{array}{rl}
              (-1)^k/4, & k = 1, 2, 3, 4 \\
              (-1)^k/2, & k = 5, 6 \\
              0, & k = 7, 8, 9, \dots
              \end{array} \right.
\end{equation}
for $k = 1$,~$2$,~$3$, \dots.
For all numerical computations associated with this example,
we can truncate to the first 20 bins,
since $\sum_{k=21}^\infty (p_0)_k < 10^{-10}$.

\subsection{Poisson model with a different alternative}
\label{Poissona}

For our fourth example, we again take
\begin{equation}
(p_0)_k = e^{-3} \, 3^{k-1}/(k-1)!
\end{equation}
for $k = 1$,~$2$,~$3$, \dots, but now take
\begin{equation}
a_k = \left\{ \begin{array}{rl}
              1, & k = 1 \\
              -1/11, & k = 2, 3, \dots, 12 \\
              0, & k = 13, 14, 15, \dots
              \end{array} \right.
\end{equation}
for $k = 1$,~$2$,~$3$, \dots.
For all numerical computations associated with this example,
we can truncate to the first 20 bins,
since $\sum_{k=21}^\infty (p_0)_k < 10^{-10}$.

\begin{remark}
The right-hand side of~(\ref{stability}) is
8.233 for Subsection~\ref{uniform},
2.443 for Subsection~\ref{nonuniform}, and
24.05 for Subsection~\ref{Poisson}.
As discussed in Remark~\ref{stabrem},
roundoff errors in the numerical evaluation of~(\ref{contoured})
are therefore guaranteed to be negligible
for the standard floating-point arithmetic
(the mantissa in the standard, double-precision arithmetic has a dynamic range
of about $5 \cdot 10^{15} \gg 24.05$).
The right-hand side of~(\ref{stability}) is
$1.478 \cdot 10^{16}$ for Subsection~\ref{Poissona},
so we used~(\ref{nocontour}) rather than~(\ref{contoured})
for the last example (Remark~\ref{auxiliary} explains why).
\end{remark}

We used Fortran~77 and ran all examples on one core
of a 2.2~GHz Intel Core~2 Duo microprocessor with 2~MB of L2 cache.
Our code is compliant with the IEEE double-precision standard
(so that the mantissas of variables have approximately one bit of precision
less than 16 digits, yielding a relative precision
of about $2 \cdot 10^{-16}$).
We diagonalized the matrix $B$ defined in~(\ref{diagonalizer})
using the Jacobi algorithm
(see, for example, Chapter~8 of~\cite{golub-van_loan}),
not taking advantage of Remark~\ref{faster};
explicitly forming the entries of the matrix $B$
defined in~(\ref{diagonalizer}) can incur a numerical error
of at most the machine precision (about $2 \cdot 10^{-16}$)
times $\max_{1 \le k \le m} (p_0)_k / \min_{1 \le k \le m} (p_0)_k$,
yielding 6-digit accuracy or better for all our examples.
A future article will exploit the interlacing properties of eigenvalues,
following~\cite{gu-eisenstat94}, to obtain higher precision.
Of course, even 4-digit precision would suffice
for most statistical applications;
however, modern computers can produce high accuracy very fast,
as the examples in this section illustrate.

\begin{figure}
\begin{center}
\rotatebox{-90}{\scalebox{.6666666666}{\includegraphics{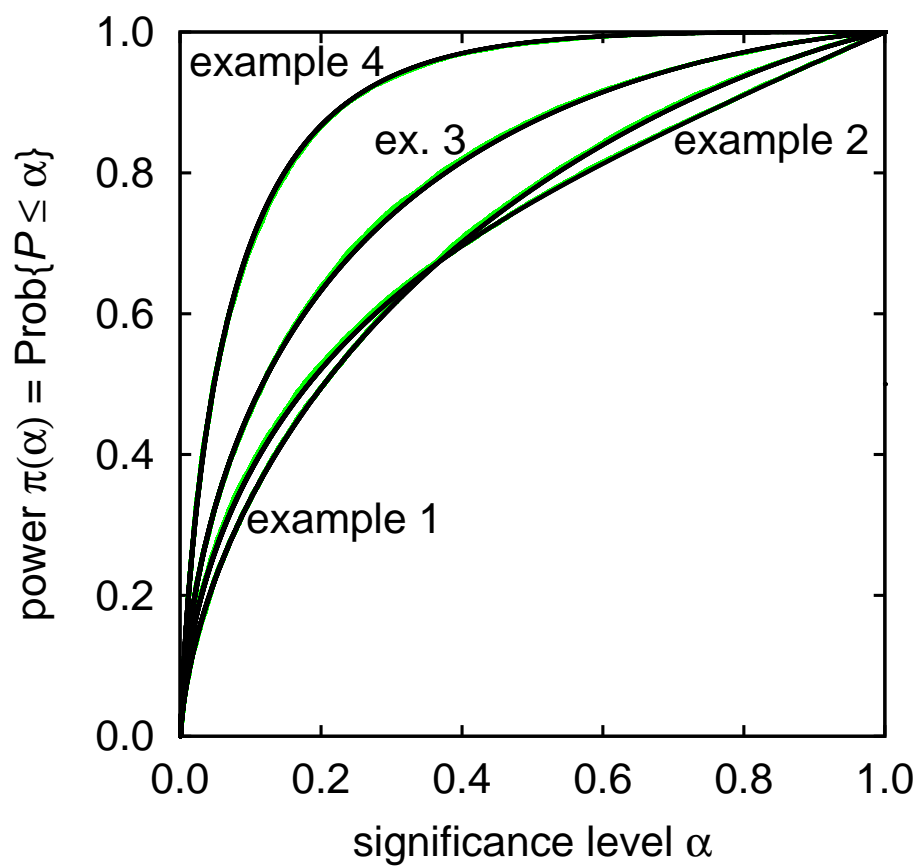}}}
\vspace{.1in}
\caption{Cumulative distribution functions of the P-values $P$
         for draws from the alternative distributions
         defined in Subsections~\ref{uniform}--\ref{Poissona}}
\label{fig}
\end{center}
\end{figure}

\begin{table}
\caption{Computational costs}
\label{costs}
\begin{center}
\begin{tabular}{c|cccc}
          & $m$ & $q_0$ & $q_a$ &   $t$ \\\hline
example 1 &  10 &   230 &   230 & 0.006 \\
example 2 & 100 &   530 &   550 & 0.090 \\
example 3 &  20 &   250 &   330 & 0.013 \\
example 4 &  20 &   350 &   350 & 0.010
\end{tabular}
\end{center}
\end{table}

\section*{Acknowledgements}

We would like to thank Jim Berger, Tony Cai, Jianqing Fan, Andrew Gelman,
Peter W. Jones, Ron Peled, Vladimir Rokhlin, and Rachel Ward
for many helpful discussions.

\newpage

\bibliographystyle{imsart-nameyear.bst}
\bibliography{stat}

\end{document}